\documentclass[conference,a4paper]{IEEEtran}
\IEEEoverridecommandlockouts
\usepackage{stfloats,color}
\usepackage{amsfonts}
\usepackage[cmex10]{amsmath}
\usepackage{graphicx}
\usepackage{setspace}
\usepackage{cite}
\usepackage{array}
\usepackage{algorithmic}
\usepackage{mdwmath}
\usepackage{mdwtab}
\usepackage[center]{caption}
\usepackage[font=footnotesize]{subfig}
\usepackage{fixltx2e}
\usepackage{url}
\usepackage{times}
\usepackage{epsfig}
\usepackage{latexsym}
\usepackage{epstopdf}
\usepackage{verbatim}
\usepackage{units}
\usepackage{amsthm}
\usepackage{mdwlist}

\theoremstyle{remark}
\newtheorem{theorem}{Theorem}

\newtheorem{proposition}[theorem]{Proposition}

\graphicspath{{./figures/}}

\begin{document}
\title{Cooperative Access in Cognitive Radio Networks: \\ Stable Throughput and Delay Tradeoffs 
\thanks{This paper was made possible by a NPRP grant 4-1034-2-385 from the
Qatar National Research Fund (a member of The Qatar Foundation). The
statements made herein are solely the responsibility of the authors.}
}

\author{\IEEEauthorblockN{Mahmoud Ashour$^\dag$, Amr A. El-Sherif$~^\P$, Tamer ElBatt$^{*\ddag}$ and Amr Mohamed$^\dag$}   

\IEEEauthorblockA{$^\dag$ Computer Science and Engineering Department, Qatar University, Doha, Qatar.\\ 
$^\P$ Dept. of Electrical Engineering, Alexandria University, Alexandria 21544, Egypt.\\
$^*$ Wireless Intelligent Networks Center (WINC), Nile University, Giza, Egypt.\\
$^\ddag$ EECE Dept., Faculty of Engineering, Cairo University, Giza, Egypt.
}
\{\tt m.ashour@qu.edu.qa, amr.elsherif, telbatt, amrm@ieee.org\}
}

\maketitle

\begin{abstract}
In this paper, we study and analyze fundamental throughput-delay tradeoffs in cooperative multiple access for cognitive radio systems. We focus on the class of randomized cooperative policies, whereby the secondary user (SU) serves either the queue of its own data or the queue of the primary user (PU) relayed data with certain service probabilities. The proposed policy opens room for trading the PU delay for enhanced SU delay. Towards this objective, stability conditions for the queues involved in the system are derived. Furthermore, a moment generating function approach is employed to derive closed-form expressions for the average delay encountered by the packets of both users. Results reveal that cooperation expands the stable throughput region of the system and significantly reduces the delay at both users. Moreover, we quantify the gain obtained in terms of the SU delay under the proposed policy, over conventional relaying that gives strict priority to the relay queue. 

\end{abstract}


\IEEEpeerreviewmaketitle

\vspace{-10pt}

\section{Introduction}
\IEEEPARstart{S} \small{pectrum} scarcity coupled with its under-utilization \cite{FCC} stimulated the introduction of cognitive radios \cite{Mitola,Haykin}. The main idea of cognitive radios resides in introducing secondary users (SUs) capable of sensing the spectrum and exploiting periods in which primary users (PUs) are idle. Due to the broadcast nature of wireless channels, cooperative communication in wireless networks has been widely investigated \cite{Tse,Kramer}. Incorporating cooperation into cognitive radio networks, the SUs not only seek idle time slots to transmit their own data, but they may also relay the PUs' packets. Thus, cooperation in cognitive radio networks can be viewed as a win-win situation. The SUs help the PUs deliver their packets to the destination. This helps in fulfilling the demand of the PUs and, hence, increases the availability of slots in which SUs can transmit their own packets.

Cooperative communication can be also viewed as a way of implementing the notion of spatial diversity. Analogous to using multiple antennas to achieve spatial diversity in single communication links \cite{Foschini,Telatar}, the resources of multiple nodes can be exploited to induce a similar effect. Many works addressed cooperative communication from a physical layer perspective, such as \cite{Tse,Kramer,Gamal}. However, we are interested in the implementation of cooperation at higher network layers. For instance, in \cite{simeone2008spectrum}, the PU leases its own bandwidth for a fraction of time to a secondary network in exchange of appropriate gains attributed to cooperation. 
In \cite{Ephremedis}, protocol-level cooperation is implemented among $N$ nodes in a wireless network, whereby each node is a source and a prospective relay at the same time. In \cite{Sadek}, two protocols are developed to implement cooperation in a system of $M$ source terminals, a single destination, and a single cognitive relay. While in \cite{Krikidis}, multiple protocols which allow cooperation between a PU and a set of SUs are analyzed.

Perhaps the closest to our work is \cite{Ephremedis} which presents delay analysis for a cognitive relaying scenario in which full priority is given to the relay queue. In this paper, unlike \cite{Ephremedis}, our prime objective is to develop a mathematical framework for the class of randomized cooperative policies that open room for accommodating cognitive radio systems supporting real-time applications. Towards this objective, we propose and analyze a tunable randomized cooperative policy, whereby the SU serves either the queue of its own data with probability (w.p.) $a$, or the relay queue w.p. $(1-a)$. The proposed policy is shown to enhance the SU delay at the expense of a slight degradation in the PU delay. The significance of the proposed policy lies in its tunability, whereby a variety of objectives could be realized via performing constrained optimizations over the service probability $a$.  
The main contributions of this work are summarized as follows:
\begin{enumerate}

\item We propose a randomized cooperative policy that enables trading the PU delay for enhanced SU delay, depending on the application and system QoS constraints.

\item The stable throughput region of the system is derived. 
Moreover, we derive closed-form expressions for the average delay experienced by the packets of both users. Furthermore, the effect of varying $a$ on the system's throughput and delay is thoroughly investigated. 

\item Extensive simulations are conducted to validate the obtained analytical results.
  
\item We study the fundamental throughput-delay tradeoff at both users. At any given point within the stability region of the system, the optimal value of $a$ that minimizes the average delay for the PU and SU is analytically derived.
\end{enumerate}    

The rest of this paper is organized as follows. Section \ref{system model} presents the system model along with the implemented cooperation strategy. Section \ref{stable throughput region} presents the derivation of the stable throughput region of the system. The average delay characterization of the system is provided in section \ref{delay}. Numerical results are then presented in section \ref{numerical results}. Finally, concluding remarks are drawn in section \ref{conclusion}.

\begin{figure}[t]
\begin{center}
\includegraphics[width=1\columnwidth , height=0.3\columnwidth]{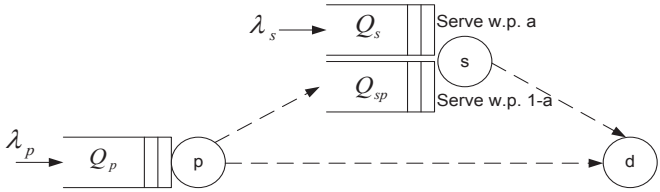}
\caption{Cognitive radio network under consideration.} \label{Fig1}
\end{center}
\vspace{-7mm}
\end{figure}
\section{System Model}\label{system model}
We consider the cognitive radio system shown in Fig. \ref{Fig1}. The system comprises a PU and a SU equipped with infinite capacity buffers, transmitting their packets to a destination $d$. Time is slotted, and the transmission of a packet takes exactly one time slot. Source burstiness is taken into account through modelling the arrivals at the PU and SU as Bernoulli processes with rates $\lambda_{p}$ and $\lambda_{s}$, respectively. The arrival processes at both users are independent of each other, and are independent and identically distributed (i.i.d) across time slots. 

We assume that the SU performs perfect sensing. Thus, the system is contention-free, since at most one user is allowed to transmit in a given slot. Hence, the only reason for packet loss is the channel outage event, which is defined as having a signal-to-noise ratio (SNR) at the receiving node below a certain threshold. Let $f_{pd}$, $f_{sd}$, and $f_{ps}$ denote the probability of no link outage between the PU and destination, the SU and destination, and the PU and SU, respectively.

\subsection{Queueing Model}
There are three queues involved in the system analysis, as shown in Fig. \ref{Fig1}. They are described as follows:
\begin{itemize}
\item $Q_{p}$: stores the packets of the PU corresponding to the external Bernoulli arrival process with rate $\lambda_{p}$.

\item $Q_{sp}$: stores the packets at the SU relayed from the PU.

\item $Q_{s}$: stores the packets of the SU corresponding to the external Bernoulli arrival process with rate $\lambda_{s}$.  
\end{itemize}

The instantaneous evolution of queue lengths is captured as
\begin{align}\label{queue evolution}
Q_{k}^{t+1}= [Q_{k}^{t} - Y_{k}^{t}]^{+} + X_{k}^{t}, ~~k \in \{ p,sp,s \}
\end{align}  
where $[x]^{+}=\text{max}(x,0)$, and $Q_{k}^{t}$ denotes the number of packets in the $k$th queue at the beginning of the $t$th time slot. The binary random variables taking values either $0$ or $1$, $Y_{k}^{t}$ and $X_{k}^{t}$, denote the departures and arrivals corresponding to the $k$th queue in the $t$th time slot, respectively. 

\subsection{Cooperation Strategy}\label{cooperation strategy}
The proposed cooperative scheme is described as follows:
\begin{enumerate}
\item The PU transmits a packet whenever $Q_{p}$ is non empty.

\item If the packet is successfully decoded by the destination, it exits the system.

\item If the packet is not successfully received by the destination, yet, successfully decoded by the SU, the packet is buffered in $Q_{sp}$ and is dropped from $Q_{p}$. It becomes the responsibility of the SU to deliver to the destination.

\item If both the destination and the SU fail to decode the packet, it is kept at $Q_{p}$ for retransmission in the next time slot.
 
\item When the PU is idle, the SU transmits a packet from $Q_{s}$ w.p. $a$, or a packet from $Q_{sp}$ w.p. ($1-a$).

\item If the packet is successfully decoded by the destination, it exits the system. Otherwise, it is kept at its queue for later retransmission. 

\end{enumerate}

It is worth noting from the description of the proposed policy that the system at hand is non work-conserving. A system is considered work-conserving if it does not idle whenever it has packets \cite{Wolff}. However, in our system, one case violates this condition, which arises when the SU detects a slot in which the PU is idle, and it randomly selects to transmit a packet from one of its queues which turns out to be empty, while the other queue is non-empty. Accordingly, the slot would go idle and be wasted despite the system having packets awaiting transmission. Clearly, this results in a degradation in the system performance. Nevertheless, we can extend it to a more flexible work-conserving version of the proposed policy that exploits the resources efficiently without the risk of wasting slots. However, its delay analysis is notoriously complex since it involves deriving the moment generating function (MGF) of the joint lengths of the three queues in the system. Thus, we resort to the non work-conserving policy for its mathematical tractability. Consequently, we derive closed-form expressions for the expected packet delay, formulate and solve, analytically, optimization problems with the objective of minimizing delay at both users.   

\section{Stable Throughput Region}\label{stable throughput region}
In this section, we characterize the stable throughput region of the system. Moreover, we distill valuable insights related to the effect of tuning $a$ on the stability region of the system.

\begin{theorem}\label{Thm1}
The stable throughput region for the system in Fig. \ref{Fig1} under the proposed randomized strategy, for a certain value of $a$, is given by
\begin{align}\label{stable throughput region for a given value of a}
\mathbf{R}=\bigg\{& (\lambda_{p},\lambda_{s}): \lambda_{s}<af_{sd} \left[ 1 - \frac{\lambda_{p}}{\mu_{p}} \right], \notag \\ 
&\text{for}~ \lambda_{p}< \frac{f_{sd}(1-a)[f_{pd} + f_{ps}(1-f_{pd})]}
{f_{sd}(1-a)+f_{ps}(1-f_{pd})}  \bigg\}
\end{align}
\end{theorem}
\begin{proof}
We use Loynes' theorem \cite{Loynes} to establish the stability of each queue. The theorem states that if the arrival and service processes of a queue are stationary, then the queue is stable if the arrival rate is strictly less than the service rate. 

\begin{itemize}
\item For $Q_{p}$ stability, the following condition must be satisfied
\begin{align}\label{Q1}
\lambda_{p} < \mu_{p}
\end{align}
where $\mu_{p}$ denotes the service rate of $Q_{p}$. A packet departs $Q_{p}$ if it is successfully decoded by at least one node, i.e., the destination or the SU. Thus, $\mu_{p}$ is given by
\begin{align}\label{mu_p}
\mu_{p}=1-(1-f_{pd})(1-f_{ps})=f_{pd} + f_{ps}(1-f_{pd})
\end{align} 
\item For $Q_{\!sp\!}$ stability, the following condition must be satisfied
\begin{align}\label{Q21}
(1-f_{pd})f_{ps}\frac{\lambda_{p}}{\mu_{p}}<\left[ 1-\frac{\lambda_{p}}{\mu_{p}} \right](1-a)f_{sd}
\end{align}
A PU's packet arrives at $Q_{sp}$ if $Q_{p}$ is not empty, which has a probability of $ \lambda_{p} / \mu_{p}$, and an outage occurs in the link between the PU and destination, which happens w.p. ($1-f_{pd}$), yet, no outage occurs in the link between the PU and SU, which happens w.p. $f_{ps}$. This explains the left hand side of (\ref{Q21}) which is the rate of packet arrivals to the SU relay queue. The right hand side represents the service rate seen by the packets of $Q_{sp}$. A packet departs $Q_{sp}$ if $Q_{p}$ is empty, $Q_{sp}$ is selected to transmit a packet, and there is no outage in the link between the SU and destination. 
Re-arranging the terms of (\ref{Q21}) yields the following condition on $\lambda_{p}$
\begin{align}\label{Q21_modified}
\lambda_{p}< \left[ \frac{f_{sd}(1-a)}{f_{sd}(1-a)+f_{ps}(1-f_{pd})} \right]\mu_{p} 
\end{align} 
Comparing (\ref{Q1}) and (\ref{Q21_modified}), it becomes clear that (\ref{Q21_modified}) provides a tighter bound on $\lambda_{p}$ due to the multiplication of $\mu_{p}$ by a term which is less than 1.

\item For $Q_{s}$ stability, the following condition must be satisfied
\begin{align}\label{Q22}
\lambda_{s}<af_{sd} \left[ 1 - \frac{\lambda_{p}}{\mu_{p}} \right]
\end{align}
Using the same rationale, a packet departs $Q_{s}$ if $Q_{p}$ is empty, $Q_{s}$ is selected to transmit a packet, and there is no outage in the link between the SU and destination. This explains the service rate seen by the packets of $Q_{s}$ given in the right hand side of (\ref{Q22}).
\end{itemize}   

Conditions in (\ref{Q21_modified}) and (\ref{Q22}) establish the result in (\ref{stable throughput region for a given value of a}). 
\end{proof} 

\begin{figure}[t]
\begin{center}
\includegraphics[width=1\columnwidth , height=0.68\columnwidth]{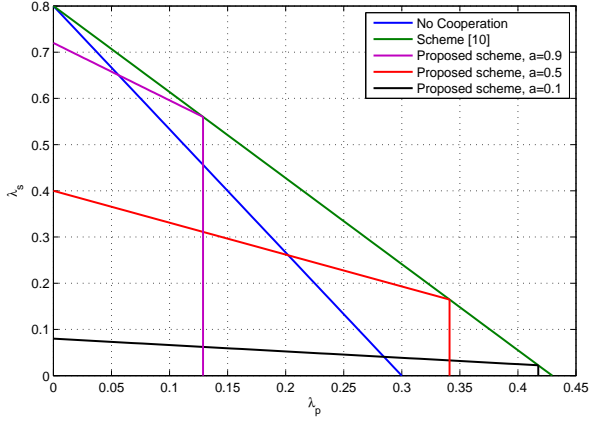}
\caption{Stable throughput region under different policies.} \label{Fig2}
\end{center}
\vspace{-7mm}
\end{figure}

\begin{proposition}\label{conjecture}
The maximum sustainable arrival rate at the PU, $\lambda_{p}$, decreases monotonically with $a$. Conversely, at a fixed $\lambda_{p}$, the maximum sustainable arrival rate at the SU, $\lambda_{s}$, increases monotonically with $a$.
\end{proposition}

\begin{proof}
From the system stability conditions, the maximum sustainable $\lambda_{p}$ for a given value of $a$ is given by (\ref{Q21_modified}). Taking the derivative of (\ref{Q21_modified}) with respect to (w.r.t.) $a$ yields 
\begin{equation}\label{lambda_p vs a}
\frac{\partial \lambda_{p}}{\partial a}=\frac{-f_{sd}f_{ps}(1-f_{pd})\mu_{p}}
{[f_{sd}(1-a)+f_{ps}(1-f_{pd})]^{2}}
\end{equation}  
Since $f_{sd}$, $f_{ps}$, $f_{pd}$, and $\mu_{p}$ are all positive numbers less than $1$, we conclude from (\ref{lambda_p vs a}) that $\frac{\partial \lambda_{p}}{\partial a}$ is negative definite. Thus, the maximum sustainable $\lambda_{p}$ monotonically decreases with $a$.

On the other hand, at a fixed $\lambda_{p}$, the maximum sustainable $\lambda_{s}$ for a given value of $a$ is given by (\ref{Q22}). Taking the derivative of (\ref{Q22}) w.r.t. $a$ yields
\begin{equation}\label{lambda_s vs a}
\frac{\partial \lambda_{s}}{\partial a}=f_{sd} \left[ 1- \frac{\lambda_{p}}{\mu_{p}} \right]
\end{equation}
The stability condition in (\ref{Q1}) guarantees that $\frac{\lambda_{p}}{\mu{p}}$ is less than $1$. Thus, $\frac{\partial \lambda_{s}}{\partial a}$ is positive definite. Therefore, at a fixed $\lambda_{p}$, the maximum sustainable $\lambda_{s}$ monotonically increases with $a$.
\end{proof}

In Fig.~\ref{Fig2}, we plot the stable throughput region of the studied system under different multiple-access policies. Hereafter, the system parameters are chosen as follows: $f_{pd}=0.3$, $f_{ps}=0.4$, and $f_{sd}=0.8$. From Fig. \ref{Fig2}, we depict the effect of the probability $a$ on the stability region of the proposed scheme. It can be realized that increasing $a$ decreases the maximum sustainable $\lambda_{p}$. On the contrary, increasing $a$ results in an increase in the maximum sustainable $\lambda_{s}$, for every feasible $\lambda_{p}$. This result is intuitive, since increasing $a$ gives more chance for transmitting the SU own packets as opposed to the PU's relayed packets. This, in turn, reduces the degree of cooperation the PU experiences from the SU and, hence, the maximum sustainable $\lambda_{p}$ decreases. On the other hand, since the SU own packets are more likely to be transmitted, the system can sustain higher values of $\lambda_{s}$.

\begin{figure}[t]
\begin{center}
\includegraphics[width=1\columnwidth , height=0.68\columnwidth]{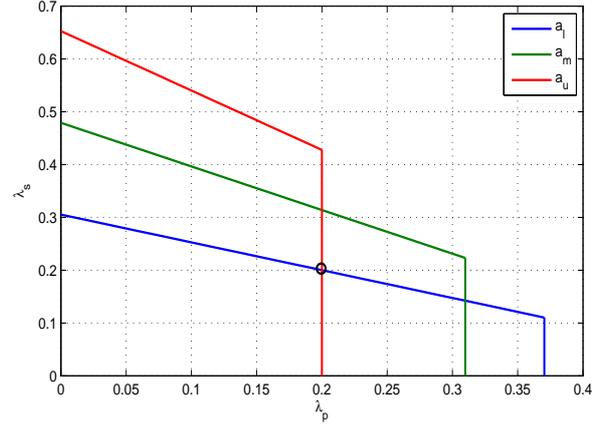}
\caption{Stable throughput region at different values of $a$ for the system to be stable at $\lambda_{p}=\lambda_{s}=0.2$.} \label{Fig3}
\end{center}
\vspace{-7mm}
\end{figure}
 
\begin{proposition}\label{proposition 1}
For any pair of arrival rates ($\lambda_{p},\lambda_{s}$) in the stable throughput region of the system, there is a bounded range of values that $a$ can take for the system to be stable at the given ($\lambda_{p},\lambda_{s}$), which is given by
\begin{equation}
\frac{\lambda_{s} \mu_{p}}{f_{sd}(\mu_{p}-\lambda_{p})} < a < 1-\frac{f_{ps}(1-f_{pd})\lambda_{p}}{f_{sd}(\mu_{p}-\lambda_{p})}
\end{equation} 
\end{proposition}

\begin{proof}
This follows directly from the equations of the stability region. Using (\ref{Q21_modified}), it can be easily shown that
\begin{equation}
f_{sd}(1-a)(\mu_{p}-\lambda_{p}) > f_{ps}(1-f_{pd})\lambda_{p}
\end{equation}
Since the system is assumed stable at ($\lambda_{p},\lambda_{s}$), then the term $(\mu_{p}-\lambda_{p})$ is positive definite. Thus, 
\begin{equation}\label{upper bound}
a < 1-\frac{f_{ps}(1-f_{pd})\lambda_{p}}{f_{sd}(\mu_{p}-\lambda_{p})}
\end{equation}
Using (\ref{Q22}), it is straightforward to show that
\begin{equation}\label{lower bound}
a > \frac{\lambda_{s} \mu_{p}}{f_{sd}(\mu_{p}-\lambda_{p})}
\end{equation}
Thus, (\ref{upper bound}) and (\ref{lower bound}) establish the result in Proposition \ref{proposition 1}.
\end{proof}

To illustrate Propositions \ref{conjecture} and \ref{proposition 1}, we compute the range of $a$ values that makes the system stable at $\lambda_{p}=\lambda_{s}=0.2$. Using (\ref{upper bound}) and (\ref{lower bound}), the upper and lower bounds on $a$, $a_{u}$ and $a_{l}$, are computed, respectively. In Fig.~\ref{Fig3}, we plot the stability region at different values of $a$, specifically at $a_{u}$, $a_{l}$, and $a_{m}=\frac{a_{u}+a_{l}}{2}$. The curves corresponding to $a_{u}$ and $a_{l}$ intersect at the point of interest, $\lambda_{p}=\lambda_{s}=0.2$. From Proposition \ref{conjecture}, we know that increasing $a$ reduces the maximum sustainable $\lambda_{p}$. This explains why at $a_{u}$ the stability region has a maximum $\lambda_{p}=0.2$. Exceeding the value of $a_{u}$ results in excluding any point with $\lambda_{p}=0.2$ from the stability region. Similarly, for the curve corresponding to $a_{l}$, it is clear that at $\lambda_{p}=0.2$, the maximum sustainable $\lambda_{s}$ is $0.2$. At a given $\lambda_{p}$, decreasing $a$ reduces the maximum sustainable $\lambda_{s}$, thus, $a_{l}$ is the minimum possible value of $a$ for which the system remains stable at $\lambda_{p}=\lambda_{s}=0.2$. Moreover, the point of interest is strictly inside the stability region for any value of $a$ in the open interval ($a_{l},a_{u}$). This is illustrated by the curve corresponding to $a_{m}$.   

\begin{theorem}
The union of the stability regions given by (\ref{stable throughput region for a given value of a}) over all possible values of $a$ is the same as that of any work-conserving cooperative scheme, e.g., the one derived in \cite{Ephremedis}, and is given by 
\begin{equation}
\lambda_{s} < f_{sd}- \left[ \frac{f_{sd}+f_{ps}(1-f_{pd})}{f_{pd}+f_{ps}(1-f_{pd})} \right] \lambda_{p}
\end{equation}
\end{theorem}

\begin{proof}
We take the union of (\ref{stable throughput region for a given value of a}) over all possible values of $a$, i.e., $a \in (0,1)$. A method used to characterize this union is proposed in \cite{Sadek} in an analogous problem. It boils down to solving a constrained optimization problem to find the maximum feasible $\lambda_{s}$ corresponding to each feasible $\lambda_{p}$. 

Herein, we use the fact that increasing $a$ increases the maximum sustainable $\lambda_{s}$ for a given $\lambda_{p}$. Thus, at a fixed $\lambda_{p}$, the maximum over all feasible $\lambda_{s}$ is achieved at the highest value of $a$ that keeps the system stable at this given $\lambda_{p}$. Moreover, the point corresponding to the maximum over all possible $\lambda_{s}$ and the given $\lambda_{p}$ defines the boundary of the stability region, and hence intuitively, there is one and only one value of $a$ that keeps the system stable at this point. Thus, the upper and lower bounds on $a$ that keep the system stable at this point coincide.
Equating (\ref{upper bound}) and (\ref{lower bound}) yields
\begin{equation}
\lambda_{s} \mu_{p}=f_{sd} \mu_{p} - (f_{sd}+f_{ps}(1-f_{pd}))\lambda_{p}
\end{equation}
Dividing both sides by $\mu_{p}$ and using (\ref{mu_p}), the maximum over all feasible $\lambda_{s}$ at a fixed $\lambda_{p}$ is given by
\begin{equation}
\lambda_{s} = f_{sd}- \left[ \frac{f_{sd}+f_{ps}(1-f_{pd})}{f_{pd}+f_{ps}(1-f_{pd})} \right] \lambda_{p}
\end{equation}      
which represents the boundary of the stability region of any work-conserving cooperative policy that is derived in \cite{Ephremedis} (Section III-A). Therefore, it has been established that, for any point ($\lambda_{p},\lambda_{s}$) that belongs to the stable throughput region of any work-conserving policy, there exists at least one value of $a$ that guarantees the system stability at this point under the non work-conserving proposed policy. Thus, the union of all stable throughput regions over all values of $a$ of the proposed policy spans the stability region derived in \cite{Ephremedis}. 
\end{proof} 
From Fig.~\ref{Fig2}, the union of all stable throughput regions of the proposed system (which matches that of \cite{Ephremedis}) strictly contains the stable throughput region achieved without cooperation.
\section{Average Delay Characterization}\label{delay}
In this section, we derive closed-form expressions for the average delay encountered by the packets of both users. 

\begin{theorem}\label{Thm2}
The average delay encountered by the packets of the PU and SU, $D_{p}$ and $D_{s}$, respectively, under the proposed randomized scheme, are given by 
\begin{equation}\label{Dp}
D_{p}= \frac{N_{p} + N_{sp}}{\lambda_{p}}
\end{equation}
\begin{equation}\label{Ds}
D_{s}= \frac{N_{s}}{\lambda_{s}} \quad 
\end{equation}
where $N_{p}$ and $N_{sp}$, the average lengths of $Q_{p}$ and $Q_{sp}$, respectively, are given by
\begin{equation}\label{pk}
N_{p}=\frac{-\lambda_{p}^{2} + \lambda_{p}}{f_{pd}+f_{ps}(1-f_{pd})-\lambda_{p}}
\end{equation} 

\begin{equation}\label{Nsp}
N_{sp}= \frac{m \lambda_{p}^{2} + n \lambda_{p}}{\alpha \lambda_{p}^{2} + \beta \lambda_{p} + \gamma}
\end{equation} 
where
\begin{align}
\begin{split}
& m = f_{ps}(1-f_{pd}) \bigg[ \frac{(1-a)f_{sd}-f_{pd}}{f_{pd}+f_{ps}(1-f_{pd})} \\
&\phantom{m = f_{ps}(1-f_{pd})\bigg[}- (1-a)f_{sd} - f_{ps}(1-f_{pd}) \bigg ] 
\end{split} \notag \\
\begin{split}
n = f_{ps}(1-f_{pd})\left[ f_{pd}+f_{ps}(1-f_{pd}) \right ]
\end{split}\notag \\
\begin{split}
\alpha = (1-a)f_{sd} + f_{ps}(1-f_{pd})
\end{split}\notag \\
\begin{split}
\beta = \left[ f_{pd}+f_{ps}(1-f_{pd}) \right] \left[ -2(1-a)f_{sd}-f_{ps}(1-f_{pd}) \right]
\end{split}\notag \\
\begin{split}
\gamma = (1-a)f_{sd} \left[ f_{pd}+f_{ps}(1-f_{pd}) \right]^{2}
\end{split}
\end{align}
and $N_{s}$, the average length of $Q_{s}$, is given by
\begin{equation}\label{Ns}
N_{s}=\frac{\lambda_{p}\lambda_{s}A + (\lambda_{s}^{2}-\lambda_{s})B(B+\lambda_{p})}{BC}
\end{equation}
where
\begin{align}
\begin{split}
A=af_{sd}[f_{pd}+f_{ps}(1-f_{pd})-1]
\end{split} \notag \\ 
\begin{split}
B=f_{pd}+f_{ps}(1-f_{pd})-\lambda_{p}
\end{split} \notag \\
\begin{split}
C=(\lambda_{s}-af_{sd})[f_{pd}+f_{ps}(1-f_{pd})] + af_{sd}\lambda_{p}
\end{split}
\end{align}
\end{theorem}

\begin{proof}
We start by computing the average delay of the packets of the SU followed by that of the PU. 

Applying Little's law on $Q_{s}$ renders $D_{s}$ exactly as given by (\ref{Ds}). Thus, it remains to calculate $N_{s}$. The dependence of the service processes at $Q_{s}$ and $Q_{sp}$ on the state of $Q_{p}$ is inherent from the concept of cognitive radios.  
It is worth noting that the non work-conserving behavior of the proposed policy makes the delay analysis mathematically tractable, since $Q_{s}$ and $Q_{sp}$ become independent. To analyze the average delays at different queues, we resort to the MGF approach \cite{33}. The MGF of the joint queue lengths $Q_{p}$ and $Q_{s}$ is defined as 
\begin{align}\label{Gxy}
G(x,y)= \lim_{t \rightarrow \infty} \mathbf{E} \left[ x^{Q_{p}^{t}} y^{Q_{s}^{t}} \right]
\end{align} 
where $\mathbf{E}$ and $\mathbf{P}$ denote the statistical expectation and the probability operators, respectively. Expanding (\ref{Gxy}), taking its derivative w.r.t. $y$ and substituting by $x=y=1$ yields
\begin{align}
G_{y}(1,1)=\lim_{t \rightarrow \infty} \displaystyle \sum_{j=0}^{\infty} j \mathbf{P}
\left[Q_{s}^{t}=j \right] = N_{s}
\end{align}
Using the queue evolution form provided by (\ref{queue evolution}), we write
\begin{align}
&\mathbf{E} \left[ x^{Q_{p}^{t+1}} y^{Q_{s}^{t+1}} \right] =
\mathbf{E} \left[ x^{(Q_{p}^{t}-Y_{p}^{t}+X_{p}^{t})} y^{(Q_{s}^{t}-Y_{s}^{t}+X_{s}^{t})} \right] 
\notag \\
&= (\lambda_{p}x+1-\lambda_{p})(\lambda_{s}y+1-\lambda_{s}) 
\mathbf{E} \left[ x^{(Q_{p}^{t}-Y_{p}^{t})} y^{(Q_{s}^{t}-Y_{s}^{t})} \right]
\end{align}
This follows from the independent arrivals at $Q_{p}$ and $Q_{s}$, that yield independent Bernoulli distributed random variables, $X_{p}^{t}$ and $X_{s}^{t}$, which produce MGFs of $(\lambda_{p}x+1-\lambda_{p})$ and $(\lambda_{s}y+1-\lambda_{s})$, respectively. Expanding the above equation, we have
\begin{align}\label{expansion}
&\mathbf{E} \left[ x^{Q_{p}^{t+1}} y^{Q_{s}^{t+1}} \right]= \qquad \qquad \qquad \qquad \qquad \qquad \qquad \notag \\ &(\lambda_{p}x+1-\lambda_{p})(\lambda_{s}y+1-\lambda_{s})
\bigg\{ \mathbf{E}[\mathbf{1}[Q_{p}^{t}=0,Q_{s}^{t}=0]] \notag \\
& +\left[ \frac{f_{pd}+f_{ps}(1-f_{pd})}{x} + (1-f_{ps})(1-f_{pd}) \right]\notag \\
& \times \mathbf{E}[x^{Q_{p}^{t}}.\mathbf{1}[Q_{p}^{t}>0,Q_{s}^{t}=0]] \notag \\
& +\left[ \frac{af_{sd}}{y} + 1 -af_{sd} \right]
\mathbf{E}[y^{Q_{s}^{t}}.\mathbf{1}[Q_{p}^{t}=0,Q_{s}^{t}>0]]\notag \\
& +\left[ \frac{f_{pd}+f_{ps}(1-f_{pd})}{x} + (1-f_{ps})(1-f_{pd}) \right]\notag \\
& \times \mathbf{E}[x^{Q_{p}^{t}}y^{Q_{s}^{t}}.\mathbf{1}[Q_{p}^{t}>0,Q_{s}^{t}>0]] \bigg\}
\end{align}
To explain the terms inside the braces of (\ref{expansion}), we analyze the $4$ possible combinations of the queue states, $Q_{p}^{t}$ and $Q_{s}^{t}$
\begin{itemize}
\item $Q_{p}^{t}=0, ~Q_{s}^{t}=0$

Since both queues are empty, $Y_{p}^{t}=Y_{s}^{t}=0$. This explains the first term in the braces in (\ref{expansion}).

\item $Q_{p}^{t}>0, ~Q_{s}^{t}=0$

Clearly, no departures occur at $Q_{s}$ since it is empty, i.e., $Y_{s}^{t}=0$. At the PU side, it transmits a packet whenever it has a non-empty queue. Thus, $Y_{p}^{t}$ is given by
\begin{equation}
    Y_{p}^{t}=
    \begin{cases}
      1, & \text{w.p.}\ f_{pd}+f_{ps}(1-f_{pd}) \\
      0, & \text{w.p.}\ (1-f_{ps})(1-f_{pd})
    \end{cases}
  \end{equation}
This states that a departure occurs at $Q_{p}$ if the packet is decoded by at least one node, either the destination or the SU. Otherwise, no departures occur and the packet remains at $Q_{p}$ for retransmission in the next slot. This gives the second term in the braces in (\ref{expansion}). 

\item $Q_{p}^{t}=0, ~Q_{s}^{t}>0$

The PU is idle, thus, $Y_{p}^{t}=0$. Then, the SU gains access to the system and transmits a packet. It randomly selects the source of this packet to be either $Q_{s}$ or $Q_{sp}$. Therefore, $Y_{s}^{t}$ is given by  
\begin{equation}
Y_{s}^{t}=
    \begin{cases}
      1, & \text{w.p.}\ af_{sd} \\
      0, & \text{w.p.}\ 1-af_{sd}
    \end{cases}
\end{equation}
This states that a departure occurs at $Q_{s}$ if it is selected to transmit, which happens w.p. $a$, and the transmitted packet is successfully decoded by the destination, which happens w.p. $f_{sd}$. Otherwise, no departures occur. This results in the third term in the braces in (\ref{expansion}).

\item $Q_{p}^{t}>0, ~Q_{s}^{t}>0$

Since the PU has the priority to transmit, the SU is silent and $Y_{s}^{t}=0$. The PU transmits a packet and $Q_{p}$ evolves exactly following the case of $Q_{p}^{t}>0, ~Q_{s}^{t}=0$ yielding the last term in the braces in (\ref{expansion}).
\end{itemize}

Taking the limit when $t \rightarrow \infty$ at both sides of (\ref{expansion}), we get
\begin{align} \label{G}
G(x,y)=(\lambda_{p}x+1-\lambda_{p})(\lambda_{s}y+1-\lambda_{s}) \notag \\
\times \frac{b(x,y)G(0,0)+c(x,y)G(0,y)}{yd(x,y)}
\end{align}
where
\begin{align} 
\begin{split}
& b(x,y)=xyaf_{sd}-xaf_{sd}
\end{split} \notag \\
\begin{split}
& c(x,y)\!\!=\!\! xaf_{sd}\!\! - \!\!y[f_{pd}\!\! + \!\! f_{ps}(1\!\! - \!\! f_{pd})]\!\! + \!\!xy[f_{sd} \!\! + \!\! f_{ps}(1 \!\! - \!\!f_{pd}) \!\! - \!\! af_{sd}]
\end{split}\notag \\
\begin{split}
& d(x,y)=x-(\lambda_{p}x+1-\lambda_{p})(\lambda_{s}y+1-\lambda_{s})\times  \\
& \phantom{d(x,y)=x-(} [f_{pd}+f_{ps}(1-f_{pd})+x(1-f_{ps})(1-f_{pd})]
\end{split}
\end{align}
From the definition of $G(x,y)$, note that
\begin{align}
\begin{split}
G(0,0)=\lim_{t \rightarrow \infty} \mathbf{E}[\mathbf{1}[Q_{p}^{t}=0,Q_{s}^{t}=0]]
\end{split}\notag \\
\begin{split}
G(x,0)=G(0,0)+\lim_{t \rightarrow \infty} \mathbf{E}[x^{Q_{p}^{t}}.\mathbf{1}[Q_{p}^{t}>0,Q_{s}^{t}=0]]
\end{split} \notag \\
\begin{split}
G(0,y)=G(0,0)+\lim_{t \rightarrow \infty} \mathbf{E}[y^{Q_{s}^{t}}.\mathbf{1}[Q_{p}^{t}=0,Q_{s}^{t}>0]]
\end{split} \notag \\
\begin{split}
& G(x,y)=G(x,0)+G(0,y)-G(0,0) \\
& \phantom{G(x,y)=} +\lim_{t \rightarrow \infty} 
\mathbf{E}[x^{Q_{p}^{t}}y^{Q_{s}^{t}}.\mathbf{1}[Q_{p}^{t}>0,Q_{s}^{t}>0]]
\end{split}
\end{align}
Along the lines of \cite{33}, $G(0,0)$ is evaluated using the normalization condition, $G(1,1)=1$, by taking the limit of (\ref{G}) when $(x,y)\rightarrow(1,1)$, which yields
\begin{equation}\label{G(0,0)}
G(0,0)\!\! = \!\! 
\frac{af_{sd}[f_{pd} \!\! + \!\! f_{ps}(1\!\! - \!\! f_{pd})\!\! - \!\! \lambda_{p}]
\!\! - \!\! \lambda_{s}[f_{pd}\!\! + \!\! f_{ps}(1\!\! - \!\! f_{pd})]}
{af_{sd}[f_{pd}+f_{ps}(1-f_{pd})]}
\end{equation}
In the derivation of (\ref{G(0,0)}), we use the fact that
\begin{equation}
G(0,1)=\lim_{t \rightarrow \infty} \mathbf{P}[Q_{p}^{t}=0]=1-\frac{\lambda_{p}}{f_{pd}+f_{ps}(1-f_{pd})}
\end{equation}
To find $N_{s}$, we solve for $G_{y}(1,1)$. We evaluate the derivative of (\ref{G}) w.r.t. $y$, then take the limit of the result when $(x,y) \rightarrow (1,1)$. Applying L'Hopital's rule twice, we obtain an equation relating $G_{y}(1,1)$ to $G_{y}(0,1)$ as
\begin{equation}\label{E1}
G_{y}(1,1)=\lambda_{s}-1+\frac{af_{sd}}{\lambda_{s}}G_{y}(0,1)
\end{equation} 
Next, we compute $\left.\frac{\partial G(y,y)}{\partial y} \right|_{y=1}$. We make use of the fact that $\left.\frac{\partial G(y,y)}{\partial y} \right|_{y=1}=N_{p}+N_{s}$, and $G_{y}(1,1)=N_{s}$. After some algebraic manipulation, we get
\begin{align}\label{E2}
G_{y}(1,1)&=
\frac{-(\lambda_{p}+\lambda_{s})^{2}+\lambda_{p}\lambda_{s}+\lambda_{p}+\lambda_{s}}
{f_{pd}+f_{ps}(1-f_{pd})-\lambda_{p}-\lambda_{s}} -N_{p} \notag \\
&+\left[ \frac{f_{pd}+f_{ps}(1-f_{pd})-af_{sd}}{f_{pd}+f_{ps}(1-f_{pd})-\lambda_{p}-\lambda_{s}} \right] G_{y}(0,1)
\end{align}
We can easily calculate $N_{p}$ by observing that $Q_{p}$ is a discrete-time $M/M/1$ queue with arrival rate $\lambda_{p}$ and service rate $\mu_{p}$. Thus, applying the Pollaczek-Khinchine formula \cite{PK}, $N_{p}$ is directly given by (\ref{pk}). Solving (\ref{E1}) and (\ref{E2}) together using the result obtained by (\ref{pk}), the term $G_{y}(0,1)$ is eliminated and $N_{s}$ is exactly given by (\ref{Ns}) in Theorem \ref{Thm2}. 

Next, we calculate the average delay of the PU's packets.
A PU's packet, if directly delivered to the destination, experiences the queueing delay at $Q_{p}$ only.
This happens w.p. $1-\epsilon=\frac{f_{pd}}{1-(1-f_{ps})(1-f_{pd})}$, which is the probability that the packet is successfully decoded by the destination given that it is dropped from $Q_{p}$. Otherwise, if the transmission through the direct link between the PU and destination fails, the packet is relayed through $Q_{sp}$ and, hence, experiences the total queueing delay at both $Q_{p}$ and $Q_{sp}$.
This happens w.p. $\epsilon$. Therefore, the average delay that a PU's packet experiences is given by
\begin{equation}\label{1}
D_{p}= (1-\epsilon)\tau_{p} + \epsilon (\tau_{p}+ \tau_{sp})= \tau_{p} + \epsilon \tau_{sp}
\end{equation}
where $\tau_{p}$ and $\tau_{sp}$ denote the average queueing delays at $Q_{p}$ and $Q_{sp}$, respectively. Since the arrival rates at $Q_{p}$ and $Q_{sp}$ are given by $\lambda_{p}$ and $\epsilon \lambda_{p}$, respectively. Then, applying Little's law yields
\begin{equation}\label{2}
\tau_{p}=N_{p}/ \lambda_{p}, \hspace{10mm} \tau_{sp}=N_{sp}/ \epsilon \lambda_{p}
\end{equation} 
Substituting (\ref{2}) in (\ref{1}) renders $D_{p}$ exactly matching (\ref{Dp}). Provided that $N_{p}$ is shown to be given by (\ref{pk}), the calculation of $D_{p}$ boils down to evaluating $N_{sp}$. As indicated earlier, the service process at $Q_{sp}$ depends on the state of $Q_{p}$, so we again employ the MGF approach to compute $N_{sp}$. Let 
$H(x,y)=\lim_{t \rightarrow \infty} \mathbf{E}[x^{Q_{p}^{t}}y^{Q_{sp}^{t}}]$ be defined as the MGF of the joint queue lengths of $Q_{p}$ and $Q_{sp}$. Using an analogous derivation employed to evaluate $G(x,y)$, we write $H(x,y)$ as
\begin{equation}
H(x,y)=(\lambda_{p}x+1-\lambda_{p})\frac{b^{'}(x,y)G(0,0)+c^{'}(x,y)G(0,y)}{yd^{'}(x,y)}
\end{equation}  
where
\begin{align}
\begin{split}
b^{'}(x,y)=x(1-a)f_{sd}(y-1) 
\end{split}\notag \\
\begin{split}
c^{'}(x,y)=x(1-a)f_{sd}-yf_{pd}-y^{2}f_{ps}(1-f_{pd}) \\
+xy[f_{pd}+f_{ps}(1-f_{pd})-(1-a)f_{sd}]
\end{split}\notag \\
\begin{split}
d^{'}(x,y)\!\! = \!\! x\!\! - \!\! (\lambda_{p}x\!\! + \!\! 1 \!\! - \!\! \lambda_{p})
[f_{pd}\!\! + \!\! yf_{ps}(1\!\! - \!\! f_{pd})\!\! + \!\! x(1\!\! - \!\! f_{ps})
(1\!\! - \!\! f_{pd})]
\end{split}
\end{align}
Following the same footsteps of the approach employed to evaluate $N_{s}$, $N_{sp}$ is shown to be given by (\ref{Nsp}). 
\end{proof}

In order to be able to solve the optimization problem formulated later, we check how $D_{p}$ and $D_{s}$ behave in response to variations in $a$.
\begin{proposition}\label{proposition 2}
Under the proposed randomized cooperative policy, if the system is stable, the average delay experienced by the packets of the PU, $D_{p}$, is a monotonically increasing function in $a$, while the average delay encountered by the packets of the SU, $D_{s}$, decreases monotonically with $a$.
\end{proposition}

\begin{proof}
The closed-form expressions of $D_{p}$ and $D_{s}$ as functions of the parameter $a$ are given in Theorem \ref{Thm2} by (\ref{Dp}) and (\ref{Ds}), respectively. Thus, to check how $D_{p}$ and $D_{s}$ behave in response to changes in $a$, we compute their derivatives w.r.t. $a$. First, we take the derivative of (\ref{Dp}) w.r.t. $a$ which yields
\begin{align}
\frac{\partial D_{p}}{\partial a}=\Omega \bigg\{(\mu_{p}^{3}-\lambda_{p}^{3})-\lambda_{p}\omega(\mu_{p}-\lambda_{p}) \bigg\}
\end{align}
where 
\begin{align}
& \Omega=\frac{f_{sd}f_{ps}(1-f_{pd})}{(\alpha \lambda_{p}^{2} + \beta \lambda_{p} + \gamma)^{2}} \notag \\
& \omega= f_{pd}+[f_{ps}(1-f_{pd})+2]\mu_{p}
\end{align}
After some manipulations, we obtain
\begin{align}\label{Dp vs a}
\frac{\partial D_{p}}{\partial a} \!\! = \!\! \Omega 
(\mu_{p}\!\! - \!\! \lambda_{p})\left[(\mu_{p} \!\! - \!\! \lambda_{p})^{2}\!\! + \!\! \lambda_{p}\mu_{p}\!\! - \!\! \lambda_{p}[f_{pd}\!\! + \!\! \mu_{p}f_{ps}(1 \!\! - \!\! f_{pd})]\right] 
\end{align}
Since $\mu_{p}<1$, we note that $\lambda_{p}[f_{pd}+\mu_{p}f_{ps}(1-f_{pd})]<\lambda_{p}[f_{pd}+f_{ps}(1-f_{pd})]=\lambda_{p}\mu_{p}$. Using this fact along with (\ref{Dp vs a}), it can be seen that
\begin{align}\label{Dp conc}
\frac{\partial D_{p}}{\partial a} > \Omega (\mu_{p}-\lambda_{p})^{3}
\end{align}
Since it has been established in the proof of Theorem \ref{Thm1} by (\ref{Q1}) that $(\mu_{p}-\lambda_{p})$ is positive definite as long as $Q_{p}$ is stable. Furthermore, it can be noticed that $\Omega$ is positive by definition. Therefore, we conclude from (\ref{Dp conc}) that the derivative of $D_{p}$ w.r.t. $a$ is positive, irrespective of the choice of $a$. This proves that $D_{p}$ is a monotonically increasing function in $a$. 

\begin{figure}[t]
\begin{center}
\includegraphics[width=1\columnwidth , height=0.68\columnwidth]{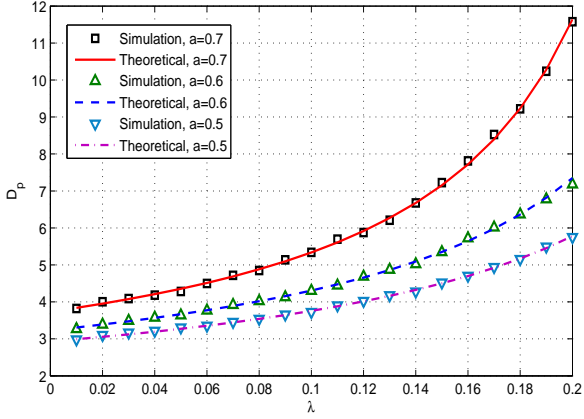}
\caption{Average delay of the PU's packets for different values of $a$.} \label{Fig4}
\end{center}
\vspace{-7mm}
\end{figure}
 
Next, we proceed with taking the derivative of (\ref{Ds}) w.r.t. $a$, which yields  
\begin{align}
\frac{\partial D_{s}}{\partial a}\!\! = \!\! f_{sd}\mu_{p}(\mu_{p}\!\! - \!\! \lambda_{p})
\bigg\{ \lambda_{s}(\lambda_{p}^{2}\!\! - \!\! \lambda_{p})\!\! + \!\!(\mu_{p}\!\! - \!\! \lambda_{p})[\lambda_{s}\mu_{p}\!\! - \!\!(\mu_{p}\!\! - \!\! \lambda_{p})]\bigg\} 
\end{align} 
Notice that $f_{sd}\mu_{p}(\mu_{p}-\lambda_{p})$ is positive definite. Thus, we focus our attention on analyzing the term in the braces. The term $(\lambda_{p}^{2}-\lambda_{p})$ is obviously negative definite since $\lambda_{p}^{2}<\lambda_{p}$. This follows from the fact that $\lambda_{p}$ is a positive number which is less than $1$. Finally, for checking the sign of $[\lambda_{s}\mu_{p}-(\mu_{p}-\lambda_{p})]$, we resort to the stability condition given by (\ref{Q22}) which states that $\lambda_{s}\mu_{p}<af_{sd}(\mu_{p}-\lambda_{p})$. Since each of $a$ and $f_{sd}$ is a  positive number less than $1$, then $af_{sd}$ is also a positive number less than $1$. Thus, it is guaranteed that $\lambda_{s}\mu_{p}<(\mu_{p}-\lambda_{p})$. This establishes that $[\lambda_{s}\mu_{p}-(\mu_{p}-\lambda_{p})]$ is negative definite. Therefore, it has been proven that $\frac{\partial D_{s}}{\partial a}$ is negative, irrespective of the choice of $a$. This, in turn, proves that $D_{s}$ is a monotonically decreasing function in $a$. 
\end{proof}

\section{Numerical Results}\label{numerical results}
In this section, we investigate the performance of the system under the proposed randomized policy. Extensive simulations are conducted to validate the closed-form expressions obtained for the average delay experienced by the packets of the PU and SU. Furthermore, we characterize and analyze a fundamental tradeoff between the average delay and the throughput at both the PU and SU. Moreover, performance comparisons to the no-cooperation scenario as well as \cite{Ephremedis} are done to show the merits of the proposed randomized scheme. 

\begin{figure}[t]
\begin{center}
\includegraphics[width=1\columnwidth , height=0.68\columnwidth]{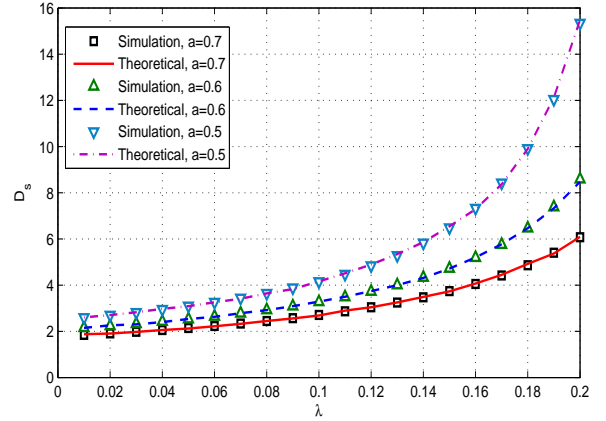}
\caption{Average delay of the SU's packets for different values of $a$.} \label{Fig5}
\end{center}
\vspace{-7mm}
\end{figure}

In Figs.~\ref{Fig4} and \ref{Fig5}, we plot the average delay experienced by the packets of the PU and the SU, respectively, versus $\lambda$, where we choose $\lambda_{p}=\lambda_{s}=\lambda$ for ease of exposition. It can be viewed that the results obtained through simulations exactly match the results of the closed-form expressions derived in Theorem \ref{Thm2}. This validates the soundness of the mathematical model and the MGF approach. Moreover, at a given $\lambda$, when $a$ increases, $D_{p}$ is shown to increase, while $D_{s}$ decreases. This matches the result stated by Proposition \ref{proposition 2}.  

Next, we characterize a fundamental tradeoff that arises between the average delay and the throughput at both the PU and SU. Intuitively, when a node needs to maintain a higher throughput, it loses in terms of the average delay encountered by its packets. Given that the system is stable, the node's throughput equals its packet arrival rate. Thus, increased throughput means injecting more packets into the system which yields a higher delay. 

In Fig. \ref{Fig6}, we illustrate the throughput-delay tradeoff at the PU. Note that, given the stability of the system, the throughput of the PU equals $\lambda_{p}$. We fix the value of $\lambda_{s}$ at $0.2$. Then, at every point $\lambda_{p}$, we formulate and solve the following optimization problem
\begin{align}
& \underset{a}{\text{minimize}}
& & D_{p} \notag \\
& \text{subject to}
& & \lambda_{p}<\left[ \frac{f_{sd}(1-a)}{f_{sd}(1-a)+f_{ps}(1-f_{pd})} \right]\mu_{p}, \notag \\
&&& \lambda_{s}<af_{sd} \left[ 1 - \frac{\lambda_{p}}{\mu_{p}} \right], \notag \\
&&& 0<a<1.
\end{align}
Thus, at every point $(\lambda_{p},\lambda_{s})$, we solve for the optimal value of $a$ that minimizes $D_{p}$, while simultaneously keeping the system stable at this point. The solution of the problem is easily done using the results obtained in Propositions \ref{proposition 1} and \ref{proposition 2}. First, we compute the feasible set, which is the set of $a$ values that guarantee the system stability at $(\lambda_{p},\lambda_{s})$. This set is given by the interval between the lower and upper bounds of $a$ defined by (\ref{lower bound}) and (\ref{upper bound}), respectively. To minimize $D_{p}$, we choose the minimum feasible value of $a$, i.e., the lower bound given by (\ref{lower bound}), since it has been established in Proposition \ref{proposition 2} that $D_{p}$ is a monotonically increasing function of $a$. Finally, we calculate $D_{p}$ using the closed-form expression obtained in Theorem \ref{Thm2} at this optimal value of $a$. Thus, in Fig. \ref{Fig6}, we show the best achievable performance of the system under the proposed randomized policy in terms of the average delay at the PU side. The tradeoff is now obvious, as the delay is shown to increase when the throughput increases. We also plot the throughput-delay curves corresponding to the no-cooperation scenario and the policy proposed in \cite{Ephremedis}. The proposed policy is shown to outperform the no-cooperation scenario, however, being outperformed by \cite{Ephremedis}. This is expected since in \cite{Ephremedis}, the priority is always given to the packets of the PU, i.e., $Q_{p}$ has the highest priority to transmit followed by $Q_{sp}$, while $Q_{s}$ never transmits except when both $Q_{p}$ and $Q_{sp}$ are empty.      

\begin{figure}[t]
\begin{center}
\includegraphics[width=1\columnwidth , height=0.68\columnwidth]{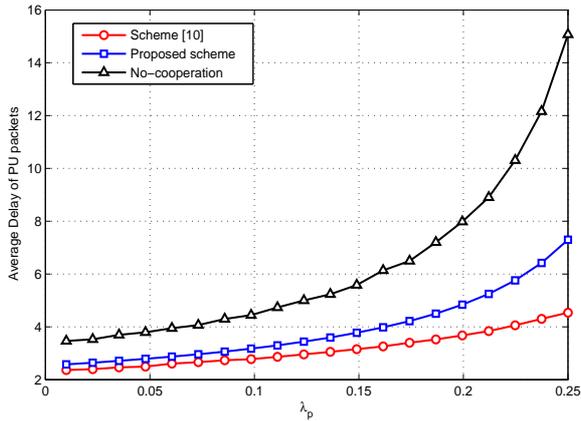}
\caption{Throughput-delay tradeoff at the PU.} \label{Fig6}
\end{center}
\vspace{-7mm}
\end{figure}

We follow the same steps at the SU side. Fixing $\lambda_{p}=0.2$, we vary the throughput of the SU, $\lambda_{s}$, to investigate its effect on $D_{s}$. For every point $(\lambda_{p},\lambda_{s})$, we minimize $D_{s}$ subject to keeping the system stable at this point. The resulting throughput-delay curves for the proposed policy as well as for \cite{Ephremedis} are shown in Fig. \ref{Fig7}. We avoided plotting the no-cooperation baseline case to have a clear view for the comparison between the plotted policies, since the no-cooperation performance is way worse than both. It can be viewed that at the SU, the best achievable performance of the system under the proposed randomized policy in terms of the average delay at the SU side, is superior to the performance of the system under the policy proposed in \cite{Ephremedis}.

\section{Conclusion}\label{conclusion}
We propose a randomized cooperative policy, whereby the SU serves either the queue of its own data or the relay queue w.p. $a$ and $(1-a)$, respectively. Results unveil performance gains of cooperation as opposed to the no-cooperation case. Cooperation is shown to expand the stable throughput region of the system as well as reducing the delay encountered by the packets of both users. Moreover, the proposed randomized policy opens room for trading the PU delay for enhanced SU delay. Furthermore, the system can be adjusted to satisfy various objectives under a set of constraints. This degree of freedom is inherited from the tunability of $a$.  

\begin{figure}[t]
\begin{center}
\includegraphics[width=1\columnwidth , height=0.68\columnwidth]{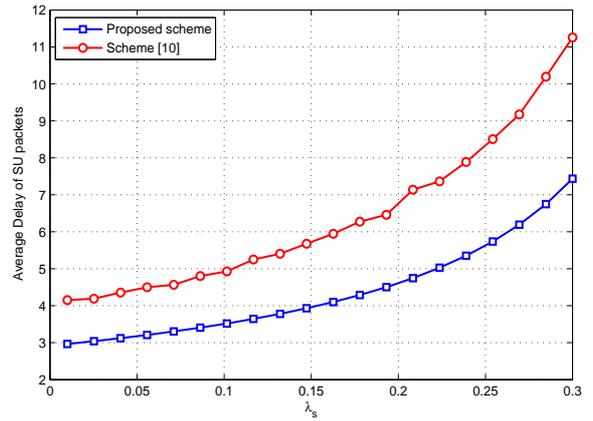}
\caption{Throughput-delay tradeoff at the SU.} \label{Fig7}
\end{center}
\vspace{-7mm}
\end{figure}
\bibliographystyle
{IEEEtran}
\bibliography{IEEEabrv,Cognitive}
\end{document}